\documentclass[journal]{IEEEtran}
\usepackage{amsmath,graphicx,amssymb,cite,algorithm,algorithmic,array,amsthm,caption,subcaption,mathtools}

\newcommand{\mat}[1] {\mathbf{#1}} 
\newcommand{\vect}[1] {\mathbf{#1}}

\newtheorem{thm}{Theorem}
\newtheorem{lem}{Lemma}

\usepackage{color}

\begin{document}
	
\title{On Probability of Support Recovery for Orthogonal Matching Pursuit Using Mutual Coherence}

\author{Ehsan~Miandji$\,^{\dagger}$,~\IEEEmembership{Student~Member,~IEEE,}%
	~Mohammad~Emadi$\,^{\dagger}$,~\IEEEmembership{Member,~IEEE,}%
	\\~Jonas~Unger,~\IEEEmembership{Member,~IEEE,}%
	~and~Ehsan~Afshari,~\IEEEmembership{Senior~Member,~IEEE}%
	\thanks{Copyright (c) 2017 IEEE. Personal use of this material is permitted. However, permission to use this material for any other purposes must be obtained from the IEEE by sending a request to pubs-permissions@ieee.org.}
	\thanks{E. Miandji and J. Unger are with the Department of Science and Technology, Link\"{o}ping University, Sweden (e-mail:  \{ehsan.miandji, jonas.unger\}@liu.se). M. Emadi is with Qualcomm Technologies Inc., San Jose, CA USA (e-mail: memadi@qti.qualcomm.com). And E. Afshari is with the Department of Electrical Engineering and Computer Science, University of Michigan, MI USA (e-mail: afshari@umich.edu).}%
	\thanks{\hspace{-8pt}$\dagger$ Equal contributer}}

\markboth{}%
{}

\maketitle

\begin{abstract}
	In this paper we present a new coherence-based performance guarantee for the Orthogonal Matching Pursuit (OMP) algorithm. A lower bound for the probability of correctly identifying the support of a sparse signal with additive white Gaussian noise is derived. Compared to previous work, the new bound takes into account the signal parameters such as dynamic range, noise variance, and sparsity. Numerical simulations show significant improvements over previous work and a closer match to empirically obtained results of the OMP algorithm. \vspace{-4pt}
\end{abstract}

\begin{IEEEkeywords}
	Compressed Sensing (CS), Sparse Recovery, Orthogonal Matching Pursuit (OMP), Mutual Coherence
\end{IEEEkeywords}

%
\IEEEpeerreviewmaketitle

\section{Introduction} \label{sec:intro}
%
\IEEEPARstart{L}et $\vect{s}\in \mathbb{R}^N$ be an unknown variable that we would like to estimate from the measurements
\vspace{-2pt}
\begin{equation} \label{eq:problem:sp}
\vect{y} = \mat{A}\vect{s}+\vect{w},
\end{equation}
where $\mat{A}\in\mathbb{R}^{M\times N}$ is a deterministic matrix and $\vect{w}\in\mathbb{R}^M$ is a noise vector, often assumed to be white Gaussian noise with mean zero and covariance $\sigma^2\mat{I}$, where $\mat{I}$ is the identity matrix. The matrix $\mat{A}$ is called a \emph{dictionary}. We consider the case when $\mat{A}$ is \emph{overcomplete}, i.e. $N>M$, hence uniqueness of the solution of \eqref{eq:problem:sp} cannot be guaranteed. However, if most elements of $\vect{s}$ are zero, we can limit the space of possible solutions, or even obtain a unique one, by solving
\vspace{-1pt}
\begin{equation} \label{eq:l0}
\hat{\vect{s}} = \underset{\vect{x}}{\mathrm{min}} \; \|\vect{x}\|_0 \;\;\; \mathrm{s.t.} \;\;\; \|\vect{y}-\mat{A}\vect{x}\|_2^2 \le \epsilon,
\end{equation}
%
where $\epsilon$ is a constant related to $\vect{w}$. The location of nonzero entries in $\vect{s}$ is known as the \emph{support} set, which we denote by $\Lambda$. In some applications, e.g. estimating the direction of arrival in antenna arrays~\cite{ref3}, correctly identifying the support is more important than accuracy of values in $\hat{\vect{s}}$. When the correct support is known, the solution of the least squares problem $\|\vect{y}-\mat{A}_{\Lambda}\vect{x}_{\Lambda}\|_2^2$ gives $\hat{\vect{s}}$, where $\mat{A}_{\Lambda}$ is formed using the columns of $\mat{A}$ indexed by $\Lambda$, see~\cite{ref1,ref30}.

Solving \eqref{eq:l0} is an NP-hard problem and several greedy algorithms have been proposed to compute an approximate solution of \eqref{eq:l0}; a few examples include Matching Pursuit (MP)~\cite{MP}, Orthogonal Matching Pursuit (OMP)~\cite{ref6}, Regularized-OMP (ROMP)~\cite{romp}, and Compressive Sampling Matching Pursuit (CoSaMP)~\cite{cosamp}. In contrast to greedy methods, convex relaxation algorithms~\cite{sparsa,lars,gpsr,spgl1} replace the $\ell_0$ pseudo-norm in \eqref{eq:l0} with an $\ell_1$ norm, leading to a convex optimization problem known as the Basis Pursuit (BP) problem~\cite{BP}. While convex relaxation methods require weaker conditions for exact recovery~\cite{ref25,ref1}, they are computationally more expensive than greedy methods, specially when $N\gg M$~\cite{fastOMP,ref4,cosamp}. 


The most important aspect of a sparse recovery algorithm is the uniqueness of the obtained solution. Mutual Coherence (MC) \cite{ref24}, cumulative coherence \cite{ref25}, the spark~\cite{DonohoSpark}, Exact Recovery Coefficient (ERC) \cite{ref23}, and Restricted Isometry Constant (RIC) \cite{ref26} are metrics proposed to evaluate the suitability of a dictionary for exact recovery. Among these metrics, RIC, spark, and ERC achieve better performance guarantees; however, computing RIC and the spark is in general NP-hard and calculating ERC is a combinatorial problem. In contrast, MC can be efficiently computed and has shown to provide acceptable performance guarantees~\cite{mc-perf,mc-perf2,ref29,ref1,ref30}. 

In this paper, we derive a new lower bound for the probability of correctly identifying the support of a sparse signal using the OMP algorithm. Our main motivation is that previous methods do not directly take into account signal parameters such as dynamic range, sparsity, and the noise characteristics in the computed probability. We will elaborate on this in section \ref{sec:motivation}, where we discuss the most recent theoretical analysis for OMP based on MC. The main result of the paper will be presented in section \ref{sec:analysis}, followed by numerical evaluation of the new performance guarantee in section \ref{sec:results}. 

\section{Motivation}	\label{sec:motivation}
The mutual coherence of a dictionary $\mat{A}$, denoted $\mu_\mathrm{max}(\mat{A})$, is the maximum absolute cross correlation of its columns~\cite{ref24}:
\begin{align}
\mu_{i,j}(\mat{A}) &= \langle \mat{A}_{i},\mat{A}_{j}\rangle, \label{eq:mc1}\\
\mu_\textrm{max}(\mat{A}) &= \underset{1\le i\neq j \le N}{\mathrm{max}}|\mu_{i,j}(\mat{A})|, \label{eq:mc2}
\end{align}
where we have assumed, as with the rest of the paper, that $\|\mat{A}_i\|_2=1$, $i\in\{1,\dots,N\}$. Apart from MC and sparsity, 
\begin{equation}
s_{\mathrm{min}} = \mathrm{min}(|\vect{s}_i|), \; \mathrm{and} \; s_{\mathrm{max}} = \mathrm{max}(|\vect{s}_i|), \;\; \forall i\in\Lambda, \label{eq:s_min_max}
\end{equation}
which define the dynamic range of the signal, also affect the performance of OMP. 
The following theorem establishes an important coherence-based performance guarantee for OMP. 
\begin{thm}[Ben-Haim et al.~\cite{ref30}] \label{thm:elad} 
	Let $\vect{y} = \mat{A}\vect{s}+\vect{w}$, where $\mat{A}\in\mathbb{R}^{M\times N}$, $\|\vect{s}\|_0=\tau$ and $\vect{w}\sim \mathcal{N}(0,\sigma^2\mat{I})$. If 
	\begin{equation} \label{eq:elad_cond}
	s_{\mathrm{min}}-(2\tau-1)\mu_{\mathrm{max}} s_{\mathrm{min}} \ge 2\beta, 
	\end{equation}
	where $\beta \triangleq \sigma \sqrt{2(1+\alpha)\log N}$ is defined for some constant $\alpha>0$, then with probability at least 
	\begin{equation} \label{eq:elad_prob}
	1 - \frac{1}{N^\alpha \sqrt{\pi (1+\alpha)\log N}},
	\end{equation}
	OMP identifies the true support, denoted $\Lambda$.
\end{thm}
The proof involves analyzing the probability event $\mathrm{Pr}\{|\langle\mat{A}_j,\vect{w}\rangle|\le\beta\}$, for some constant $\beta > 0$ and for all $j=1,\dots,N$ (see~\cite{ref30} for details). They show that with the lower bound probability of \eqref{eq:elad_prob}, the inequality $|\langle\mat{A}_j,\vect{w}\rangle|\le\beta$ holds. It is then shown that if $|\langle\mat{A}_j,\vect{w}\rangle|\le\beta$ and \eqref{eq:elad_cond} hold, then OMP identifies the correct support in each iteration. Moreover, it is assumed that the elements of the sparse vector $\vect{s}$ are deterministic variables. Hence a strong condition such as \eqref{eq:elad_cond} is required to determine if the support of $\vect{s}$ can be recovered.


Our analysis removes the condition stated in \eqref{eq:elad_cond} and introduces a probabilistic bound that depends on $N$, $\tau$, $\mu_\mathrm{max}$, $s_{\mathrm{max}}$, $s_{\mathrm{min}}$, and the signal noise. Hence we derive a probability bound that directly takes into account signal parameters and MC. Moreover, unlike \cite{ref30}, we assume that the nonzero elements of $\vect{s}$ are centered independent random variables with arbitrary distributions. This enables the derivation of a more accurate bound for the probability of exact support recovery. 


\vspace{-2pt}
\section{OMP CONVERGENCE ANALYSIS} \label{sec:analysis}
\vspace{-2pt}
In this section we present and prove the main result of the paper. Numerical results will be presented in section \ref{sec:results}. 
\begin{thm} \label{thm:ours}
Let $\vect{y} = \mat{A}\vect{s}+\vect{w}$, where $\mat{A}\in\mathbb{R}^{M\times N}$, $\tau=\|\vect{s}\|_0$ and $\vect{w}\sim \mathcal{N}(0,\sigma^2\mat{I})$. Moreover, assume that the nonzero elements of $\vect{s}$ are independent centered random variables with arbitrary distributions. Let $\lambda=\mathrm{Pr}\{|\langle\mat{A}_j,\vect{w}\rangle|\le\beta\}$, for some constant $\beta\ge 0$ and $\forall j\in\{1,\dots,N\}$. If $s_\mathrm{min}/2\ge\beta$, then OMP identifies the true support with lower bound probability 
\begin{equation} \label{eq:ours}
\lambda\left(1-2N\;\mathrm{exp}\left(\frac{-N(s_\mathrm{min}/2-\beta)^2}{2\tau^2\gamma^2+2N \gamma(s_\mathrm{min}/2-\beta)/3}\right)\right),
\end{equation}
where $\gamma=\mu_\mathrm{max}s_\mathrm{max}$. Moreover, $\lambda$ is lower bounded by
\begin{equation}	\label{eq:lambda}
1-N\sqrt{\frac{2}{\pi}}\frac{\sigma}{\beta}e^{-\beta^2/2\sigma^2}.
\end{equation}
\end{thm}
Before presenting the proof, let us compare Theorems \ref{thm:elad} and \ref{thm:ours} analytically. It is important to note that \eqref{eq:lambda} is indeed equivalent to \eqref{eq:elad_prob}. The apparent difference is only attributed to the use of $\alpha$ or $\beta$ from the definition $\beta \triangleq \sigma \sqrt{2(1+\alpha)\log N}$. For instance, using the aforementioned definition of $\beta$ on \eqref{eq:lambda} leads to \eqref{eq:elad_prob}. As a result, the second term of \eqref{eq:ours} can be interpreted as a probabilistic representation of the condition imposed by \eqref{eq:elad_cond} in Theorem \ref{thm:elad}. Moreover, because \eqref{eq:lambda} is equal to \eqref{eq:elad_prob} and the second term of \eqref{eq:ours} is in the range $[0,1]$, therefore \eqref{eq:ours} is always smaller or equal to \eqref{eq:elad_prob}. However, as it will be seen in section \ref{sec:results}, since the condition of Theorem \ref{thm:elad} in \eqref{eq:elad_cond} is not satisfied in many scenarios, our results match the empirical results more closely. Evidently, the condition $s_\mathrm{min}/2\ge\beta$ in Theorem \ref{thm:ours} is more relaxed compared to \eqref{eq:elad_cond}. Our numerical results in Section \ref{sec:results} also verify this fact.

The following lemma will provide us with the necessary tool for the proof of Theorem \ref{thm:ours}. The proof of the lemma is postponed to the Appendix. 
\begin{lem} \label{lem:1}
Define $\Gamma_j = |\langle\mat{A}_j,\mat{A}\vect{s}+\vect{w}\rangle|$, for any $j\in\{1,\dots,N\}$, where $\vect{w}\sim \mathcal{N}(0,\sigma^2\mat{I})$ and $|\langle\mat{A}_j,\vect{w}\rangle|\le\beta$. Then for some constant $\xi \ge 0$, and assuming $\xi\ge \beta$, we have
\begin{equation} \label{eq:lemma1_main}
\mathrm{Pr} \left\{ \Gamma_j \ge \xi \right\} \le 2 \; \mathrm{exp} \left(\frac{-(\xi-\beta)^2}{2(N\nu+c(\xi-\beta)/3)}\right),
\end{equation}
where 
\begin{equation}  \label{eq:lemma1_c_nu}
|\mu_{j,n}\vect{s}_n| \le c, \;\; E\left\{\mu_{j,n}^2\vect{s}_n^2\right\} \le \nu, \;\; \forall n\in \{1,\dots, N\} 
\end{equation}
\end{lem}
\noindent We can now state the proof of Theorem \ref{thm:ours}.
\begin{proof} [Proof of Theorem \ref{thm:ours}]
It was shown in \cite{ref30} that OMP identifies the true support $\Lambda$ if
\begin{equation} \label{eq:thm1_proof_1}
\underset{j\in \Lambda}{\mathrm{min}} |\langle\mat{A}_j,\mat{A}_{\Lambda}\vect{s}_{\Lambda}+\vect{w}\rangle| \ge \underset{k\notin \Lambda}{\mathrm{max}} |\langle\mat{A}_k,\mat{A}_{\Lambda}\vect{s}_{\Lambda}+\vect{w}\rangle|.
\end{equation}
\noindent The term on the left-hand side of \eqref{eq:thm1_proof_1} can be rewritten as 
\begin{align}
\underset{j\in \Lambda}{\mathrm{min}} &|\langle\mat{A}_j,\mat{A}_{\Lambda}\vect{s}_{\Lambda}+\vect{w}\rangle| \nonumber \\
&= \underset{j\in \Lambda}{\mathrm{min}}\left|\vect{s}_j + \langle\mat{A}_j,\mat{A}_{\Lambda\setminus\{j\}}\vect{s}_{\Lambda\setminus\{j\}}+\vect{w}\rangle\right|  \\
&\ge \underset{j\in \Lambda}{\mathrm{min}} \left|\vect{s}_j\right| - \underset{j\in\Lambda}{\mathrm{max}} \left|\langle\mat{A}_j,\mat{A}_{\Lambda\setminus\{j\}}\vect{s}_{\Lambda\setminus\{j\}}+\vect{w}\rangle\right|. \label{eq:thm1_proof_2}
\end{align}
From \eqref{eq:thm1_proof_1} and \eqref{eq:thm1_proof_2}, we can see that the OMP algorithm identifies the true support if
\begin{equation} \label{eq:thm1_proof_3}
\begin{cases}
\begin{aligned}
&\underset{k\notin \Lambda}{\mathrm{max}} \left\{\Gamma_k\right\} < \underset{j\in \Lambda}{\mathrm{min}} \frac{|\vect{s}_j|}{2},\\
&\underset{j\in\Lambda}{\mathrm{max}} \left|\langle\mat{A}_j,\mat{A}_{\Lambda\setminus\{j\}}\vect{s}_{\Lambda\setminus\{j\}}+\vect{w}\rangle\right| < \underset{j\in \Lambda}{\mathrm{min}} \frac{|\vect{s}_j|}{2}.
\end{aligned}
\end{cases}
\end{equation}
Using \eqref{eq:thm1_proof_3}, we can define the probability of error as
\begin{align}
\mathrm{Pr}\{\mathrm{error}\} &\le \mathrm{Pr}\left\{ \underset{j\in \Lambda}{\mathrm{max}} \left|\langle\mat{A}_j,\mat{A}_{\Lambda\setminus\{j\}}\vect{s}_{\Lambda\setminus\{j\}}+\vect{w}\rangle\right| \ge \frac{s_{\mathrm{min}}}{2} \right\} \nonumber \\
&\qquad\qquad\quad+\mathrm{Pr}\left\{\underset{k\notin \Lambda}{\mathrm{max}} \left\{\Gamma_k\right\} \ge \frac{s_{\mathrm{min}}}{2} \right\} \label{eq:thm1_proof_4} \\
&\le\sum\limits_{j\in\Lambda}^{}\mathrm{Pr}\left\{\left|\langle\mat{A}_j,\mat{A}_{\Lambda\setminus\{j\}}\vect{s}_{\Lambda\setminus\{j\}}+\vect{w}\rangle\right| \ge \frac{s_{\mathrm{min}}}{2} \right\} \nonumber \\ &\qquad\qquad\quad+\sum\limits_{k\notin\Lambda}^{}\mathrm{Pr}\left\{\Gamma_k \ge \frac{s_{\mathrm{min}}}{2}\right\}. \label{eq:thm1_proof_5}	
\end{align}
For the first term on the right-hand side of \eqref{eq:thm1_proof_5}, excluding the summation over the indices in $\Lambda$, from Lemma \ref{lem:1} we have
\begin{multline}	\label{eq:thm1_proof_6}
\underset{j\in\Lambda}{\mathrm{Pr}} \left\{\left|\langle\mat{A}_j,  \mat{A}_{\Lambda\setminus\{j\}}\vect{s}_{\Lambda\setminus\{j\}} + \vect{w}\rangle\right| \ge \frac{s_{\mathrm{min}}}{2} \right\} \\
\le \underbrace{2 \; \mathrm{exp} \left(\frac{-\rho^2}{2((\tau-1)\nu+ c\rho/3)}\right)}_{\mathrm{P}_1}, 
\end{multline}
where $\rho=s_\mathrm{min}/2-\beta$ is defined for notational brevity. Note that the dictionary $\mat{A}$ in \eqref{eq:thm1_proof_6} is supported on $\Lambda\setminus\{j\}$, i.e. all the indices in the true support excluding $j$. Therefore the term $(\tau-1)$, instead of $N$, appears in the denominator of \eqref{eq:thm1_proof_6}. Similarly, for the second term of \eqref{eq:thm1_proof_5} we have
\begin{align} 	\label{eq:thm1_proof_7}
\underset{k\notin\Lambda}{\mathrm{Pr}} \left\{\Gamma_k \ge \frac{s_{\mathrm{min}}}{2}\right\} \le \underbrace{2 \; \mathrm{exp} \left(\frac{-\rho^2}{2(\tau\nu+c\rho/3)}\right)}_{\mathrm{P}_2}.
\end{align}
\noindent Substituting \eqref{eq:thm1_proof_6} and \eqref{eq:thm1_proof_7} into \eqref{eq:thm1_proof_5} yields

\begin{equation} \label{eq:thm1_proof_8}
\mathrm{Pr}\{\mathrm{error}\} \le \tau\mathrm{P}_1 + (N-\tau)\mathrm{P}_2 \le N\mathrm{P}_2, 
\end{equation}

\noindent where the last inequality follows since $\mathrm{P}_2>\mathrm{P}_1$. 

Moreover, for the upper bounds $c$ and $\nu$ in \eqref{eq:lemma1_c_nu} we have
\begin{align}
|\mu_{j,n}\vect{s}_n| &\le \mu_\mathrm{max}s_\mathrm{max}, \label{eq:thm1_proof_9} \\ 
\mathrm{E}\left\{\mu_{j,n}^2\vect{s}_n^2\right\} &\le \frac{1}{N} \sum_{n=1}^{N}\mu_{\mathrm{max}}^2 \mathrm{E}\{\vect{s}_n^2\} \le \frac{\tau}{N} s_{\mathrm{max}}^2\mu_{\mathrm{max}}^2, \label{eq:thm1_proof_10}
\end{align}
Combining \eqref{eq:thm1_proof_9} and \eqref{eq:thm1_proof_10} with \eqref{eq:thm1_proof_8}, the following is obtained
\begin{equation}	\label{eq:thm1_proof_11}
\mathrm{Pr}\{\mathrm{error}\} \le 2N \; \mathrm{exp} \left(\frac{-N\rho^2}{2\tau^2\gamma^2+2N \gamma\rho/3}\right),  
\end{equation}
\noindent where we have defined $\gamma=\mu_\mathrm{max}s_\mathrm{max}$ for notational brevity. 

So far we have assumed that $|\langle\mat{A}_j,\vect{w}\rangle| \le \beta$, $\forall j$. Therefore, the probability of success is the joint probability of $\mathrm{Pr}\left\{|\langle\mat{A}_j,\vect{w}\rangle| \le \beta \right\}$ and the inverse of \eqref{eq:thm1_proof_11}. For the former, a lower bound was formulated in~\cite{ref30} as follows
\begin{equation}	 \label{eq:thm1_proof_12}
\mathrm{Pr}\left\{|\langle\mat{A}_j,\vect{w}\rangle| \le \beta \right\} \ge 1 - \underbrace{\sqrt{\frac{2}{\pi}}\frac{\sigma}{\beta}e^{-\beta^2/2\sigma^2}}_{P_3}.
\end{equation}
Since $|\langle\mat{A}_j,\vect{w}\rangle| \le \beta$ should hold $\forall j\in\{1,\dots, N\}$, we have
\begin{equation}	\label{eq:thm1_proof_13}
\underset{j=1,\dots,N}{\mathrm{Pr}}\left\{|\langle\mat{A}_j,\vect{w}\rangle| \le \beta \right\} \ge (1-P_3)^N \ge 1-NP_3.
\end{equation}
Inverting the probability event in \eqref{eq:thm1_proof_11} and multiplying by the lower bound in \eqref{eq:thm1_proof_13} yields \eqref{eq:ours}, which completes our proof. 
\end{proof}

\vspace{-5pt}
\section{Numerical Results} \label{sec:results}
\vspace{-2pt}
In this section we compare numerical results of Theorem \ref{thm:elad} (Ben-Haim et al.~\cite{ref30}), and Theorem \ref{thm:ours} (proposed herein) with the empirical results of OMP. Indeed we only consider probability of successful recovery of the support. An upper bound for the MSE of the oracle estimator has been previously established, see e.g. Theorem 5.1 in~\cite{ref1} or Lemma 4 in~\cite{ref30}. The oracle estimator knows the support of the signal, \emph{a priori}. 

\begin{figure*}[t] 
	\centering
	
	\begin{minipage}[t]{0.329\textwidth}
	\centering
	\captionsetup[subfigure]{aboveskip=1pt,belowskip=2pt}
	\subcaptionbox{$M=1024$, $s_\mathrm{min}=0.5$, $s_\mathrm{max}=1$,\label{tau-pr-1024}}{\includegraphics[width=0.95\linewidth]{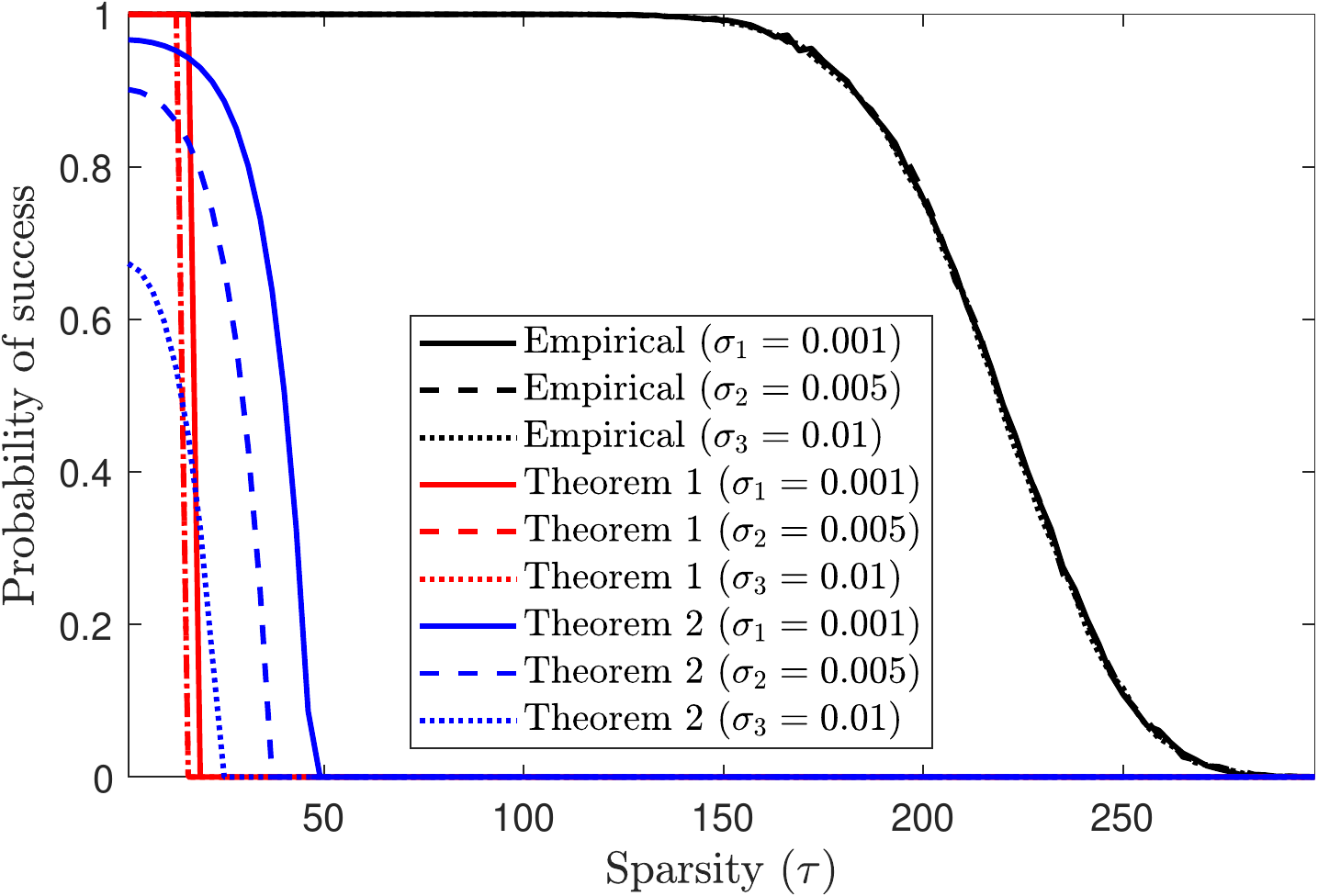}} 
	\subcaptionbox{$M=2048$, $s_\mathrm{min}=0.5$, $s_\mathrm{max}=1$,\label{tau-pr-2048}}{\includegraphics[width=0.95\linewidth]{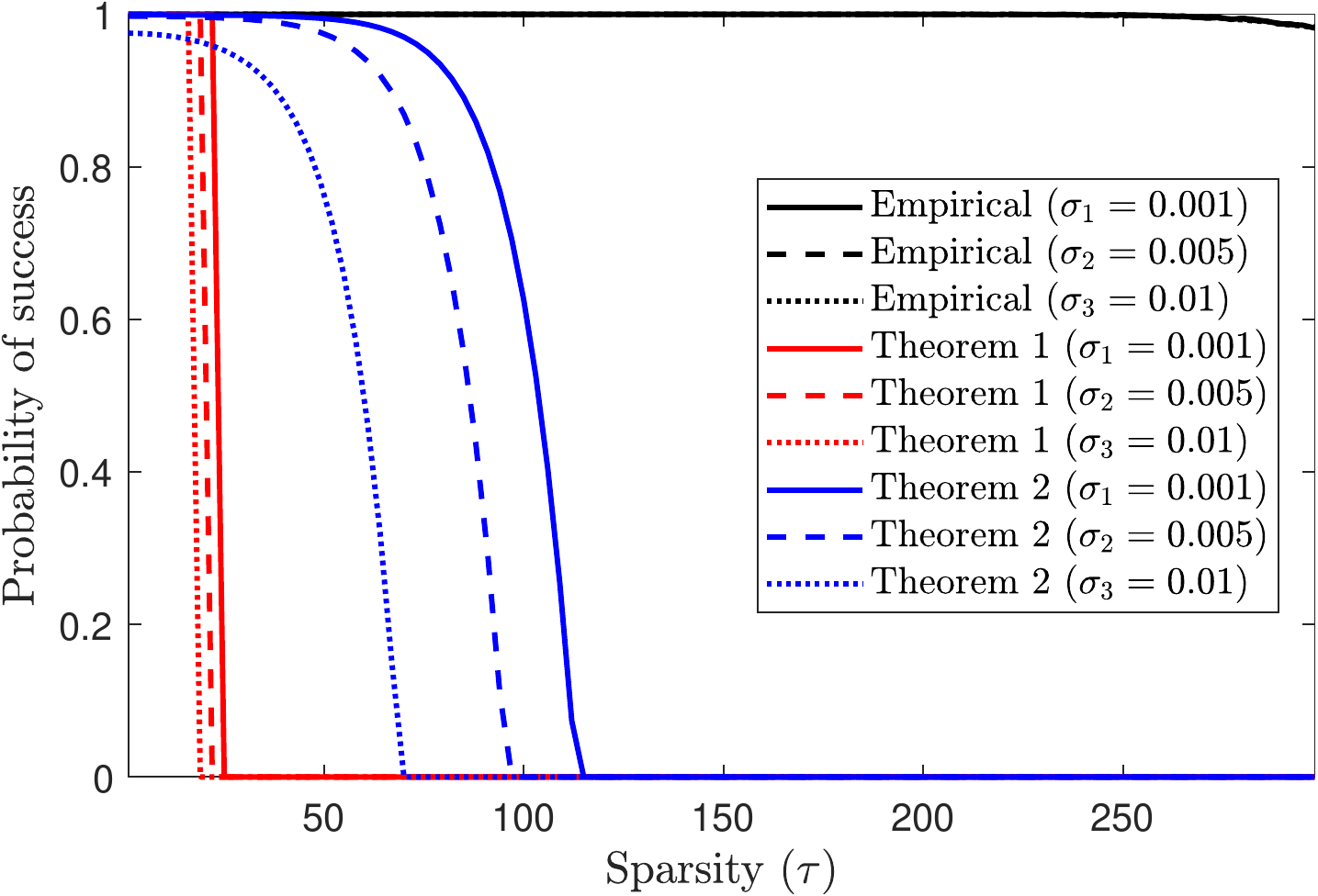}}
	\subcaptionbox{$M=4096$, $s_\mathrm{min}=0.5$, $s_\mathrm{max}=1$,\label{tau-pr-4096}}{\includegraphics[width=0.95\linewidth]{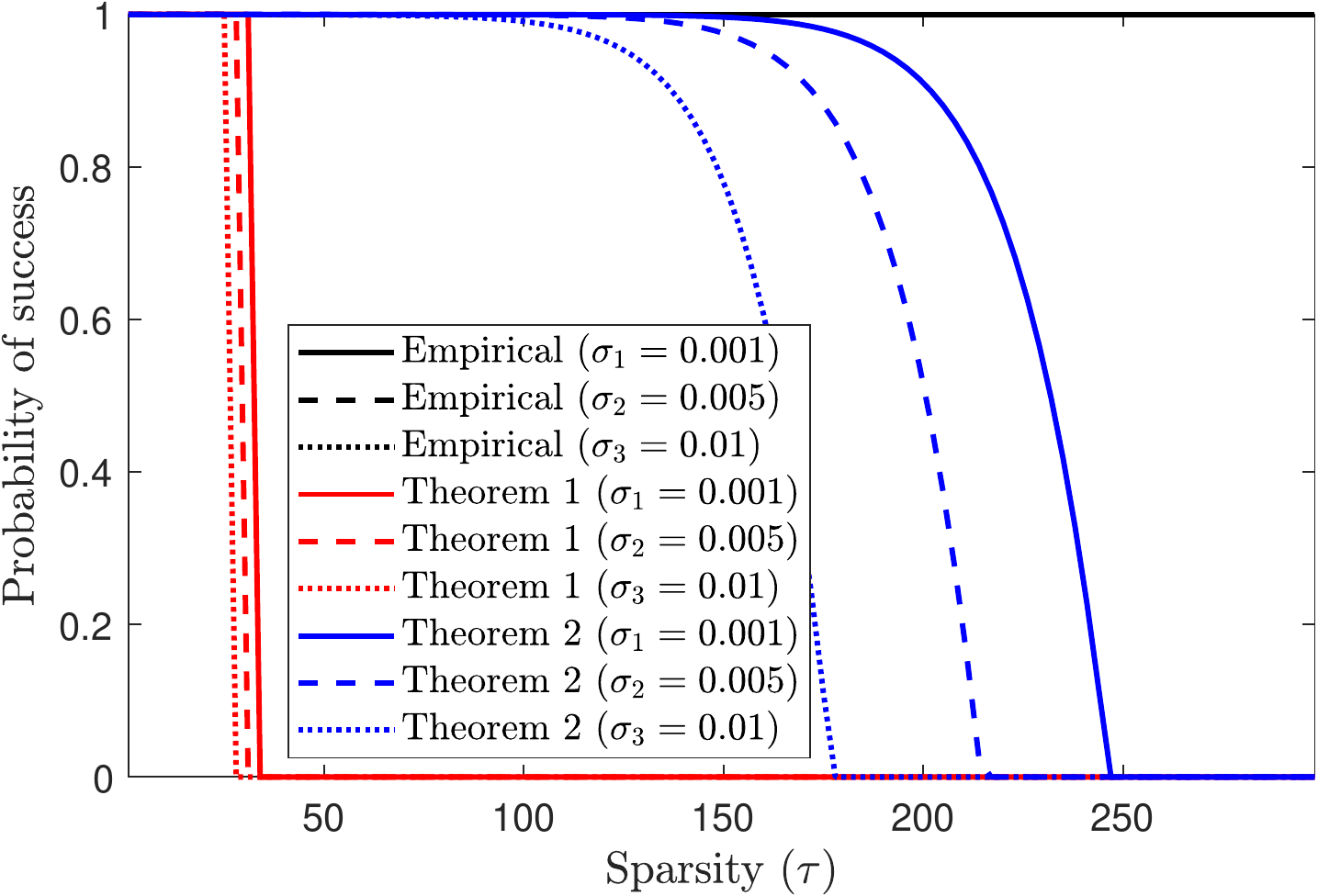}}
	\end{minipage}
	\begin{minipage}[t]{0.329\textwidth}
	\centering
	\captionsetup[subfigure]{aboveskip=1pt,belowskip=2pt}
	\subcaptionbox{$M=1024$, $s_\mathrm{max}=1$, $\sigma=0.01$\label{smin-pr-1024}}{\includegraphics[width=0.95\linewidth]{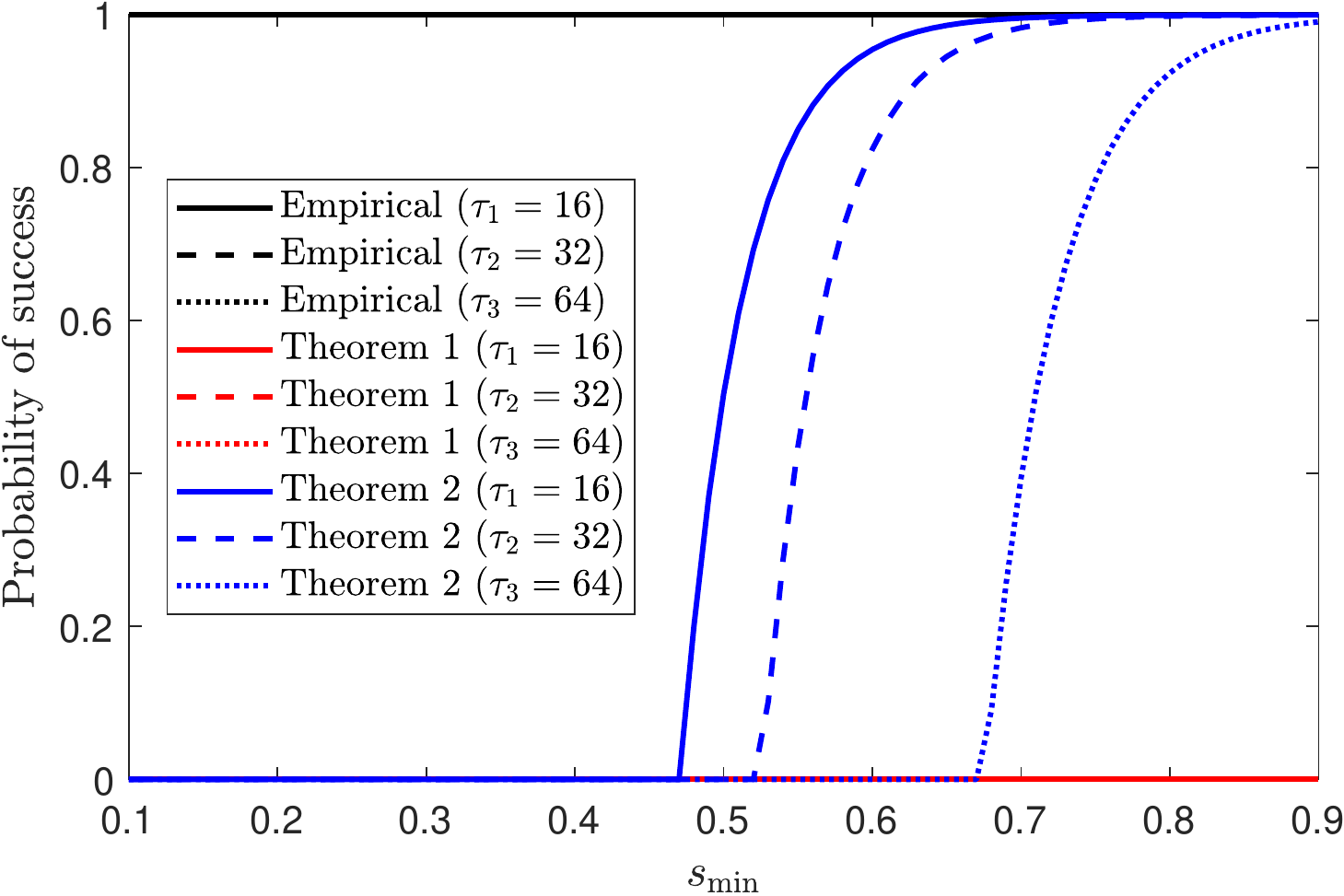}} 
	\subcaptionbox{$M=2048$, $s_\mathrm{max}=1$, $\sigma=0.01$\label{smin-pr-2048}}{\includegraphics[width=0.95\linewidth]{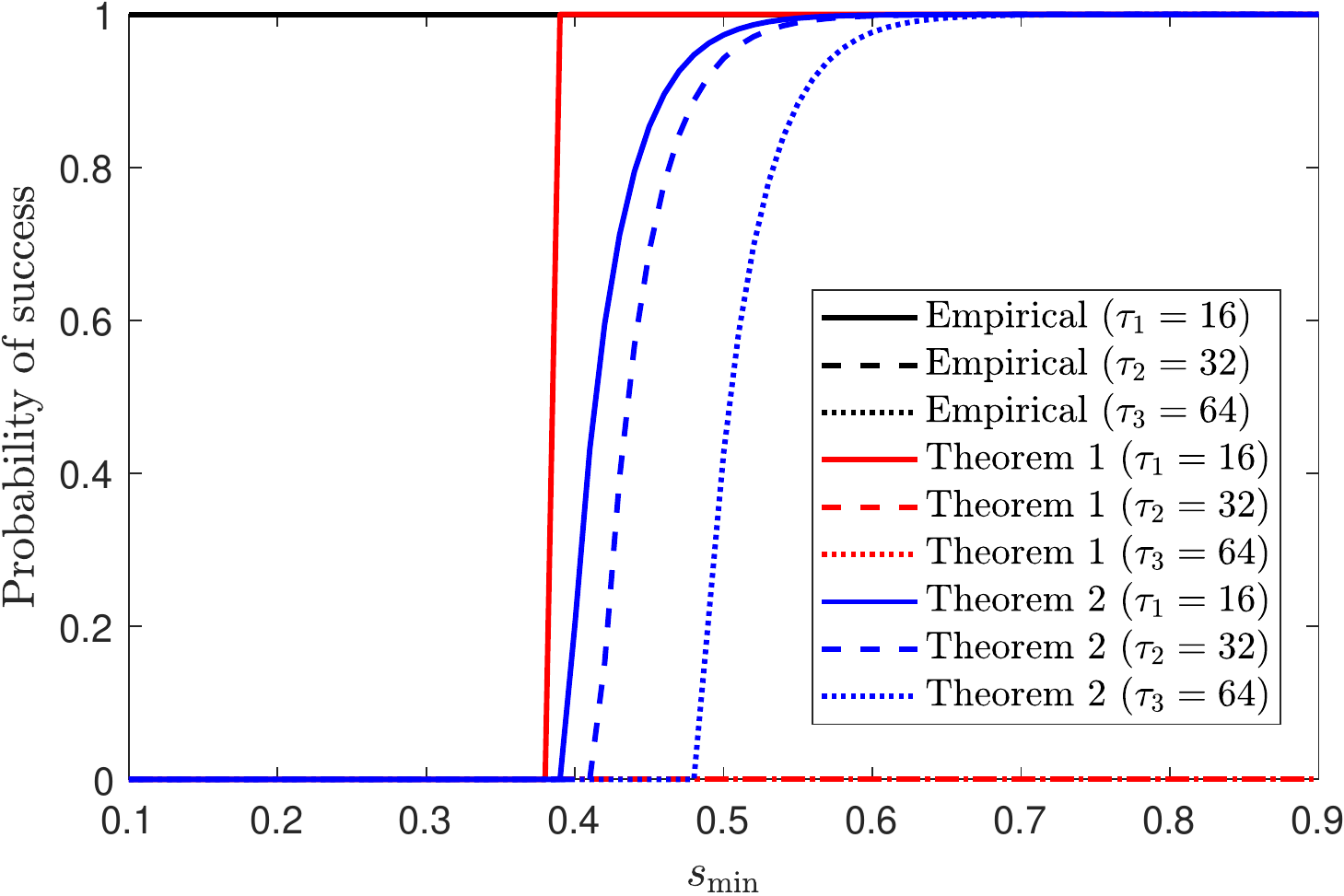}}
	\subcaptionbox{$M=4096$, $s_\mathrm{max}=1$, $\sigma=0.01$\label{smin-pr-4096}}{\includegraphics[width=0.95\linewidth]{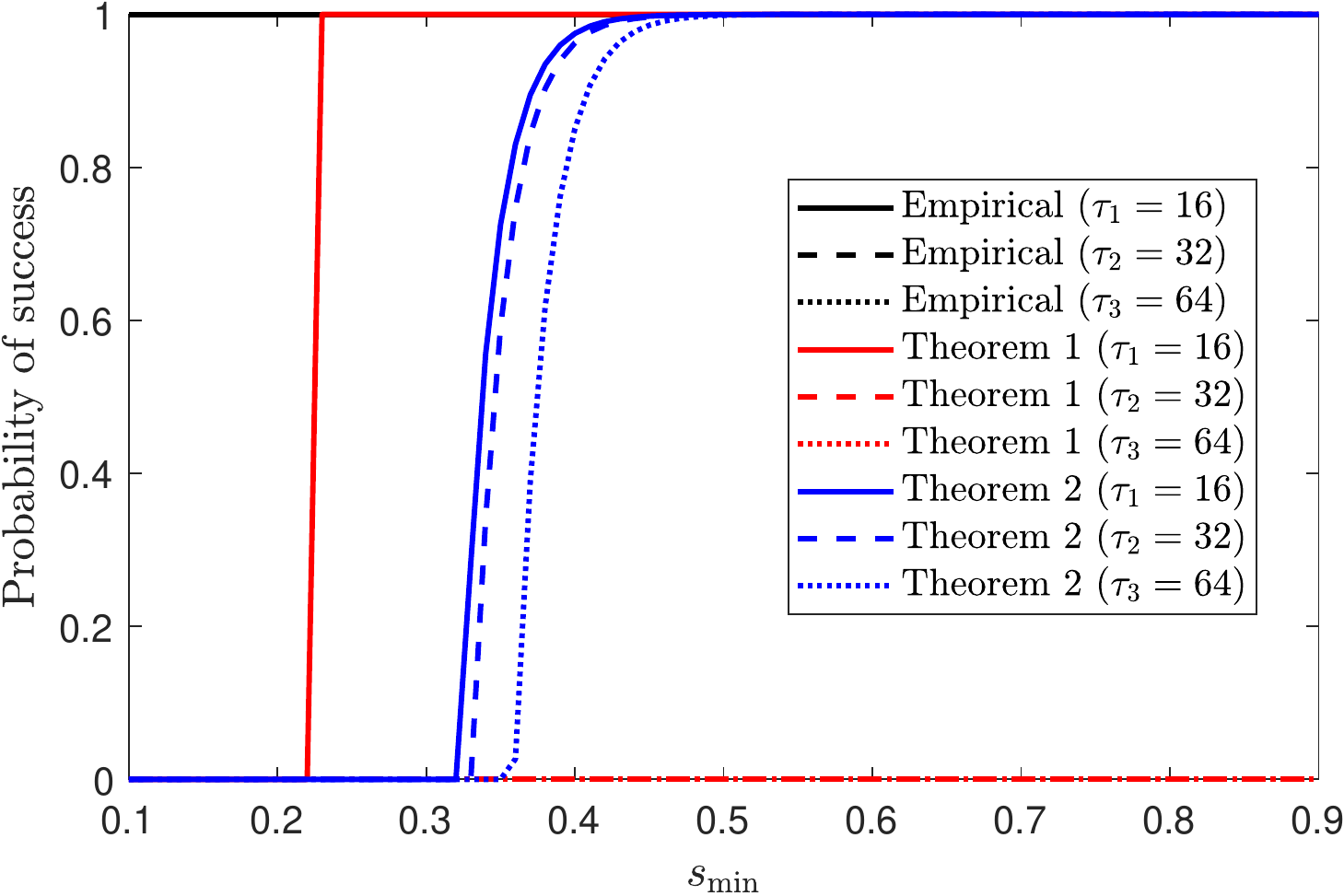}}
	\end{minipage}
	\begin{minipage}[t]{0.329\textwidth}
	\centering
	\captionsetup[subfigure]{aboveskip=1pt,belowskip=2pt}
	\subcaptionbox{$M=1024$, $s_\mathrm{min}=0.5$, $s_\mathrm{max}=1$,\label{sigma-pr-1024}}{\includegraphics[width=0.95\linewidth]{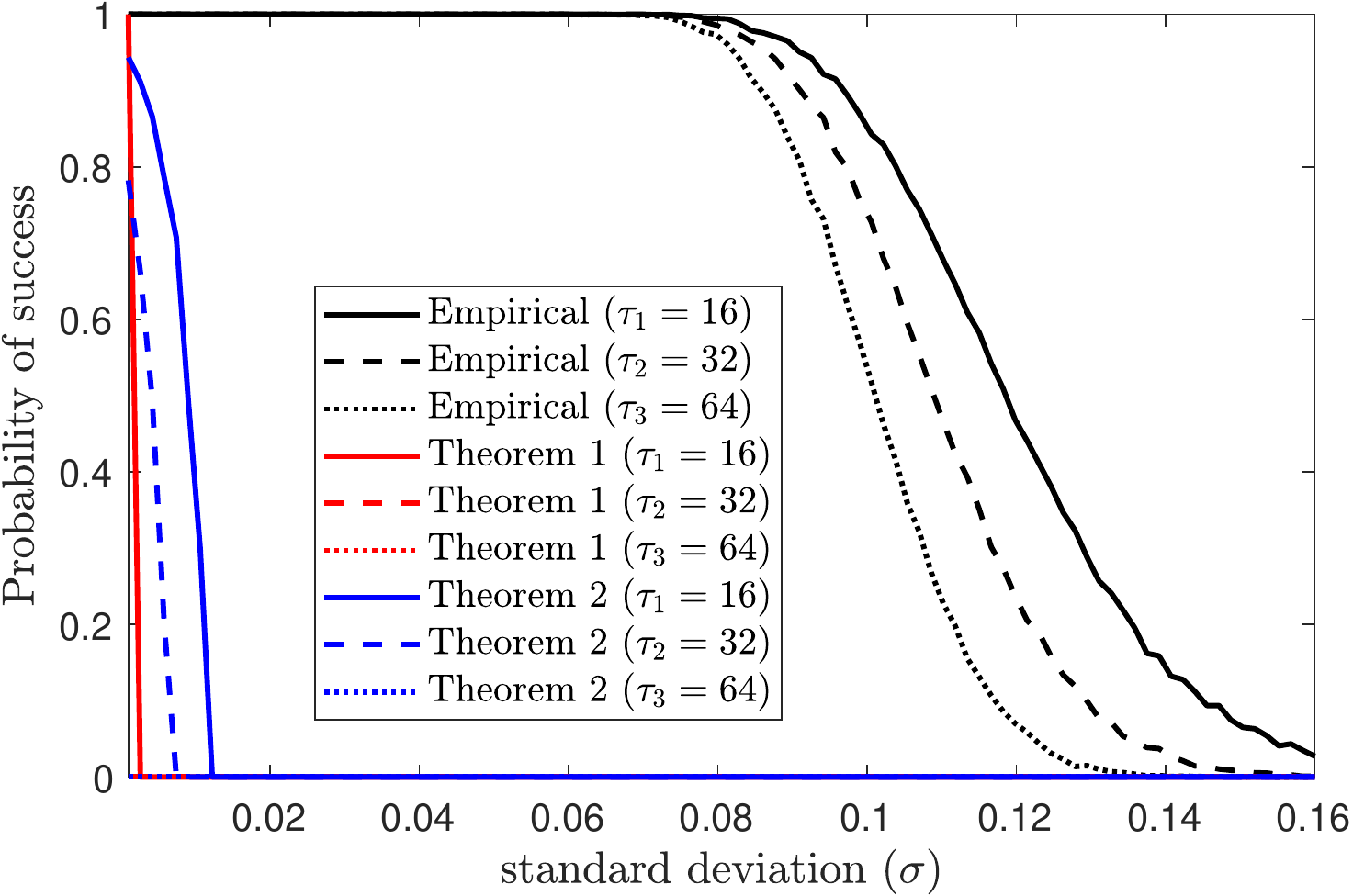}} 
	\subcaptionbox{$M=2048$, $s_\mathrm{min}=0.5$, $s_\mathrm{max}=1$,\label{sigma-pr-2048}}{\includegraphics[width=0.95\linewidth]{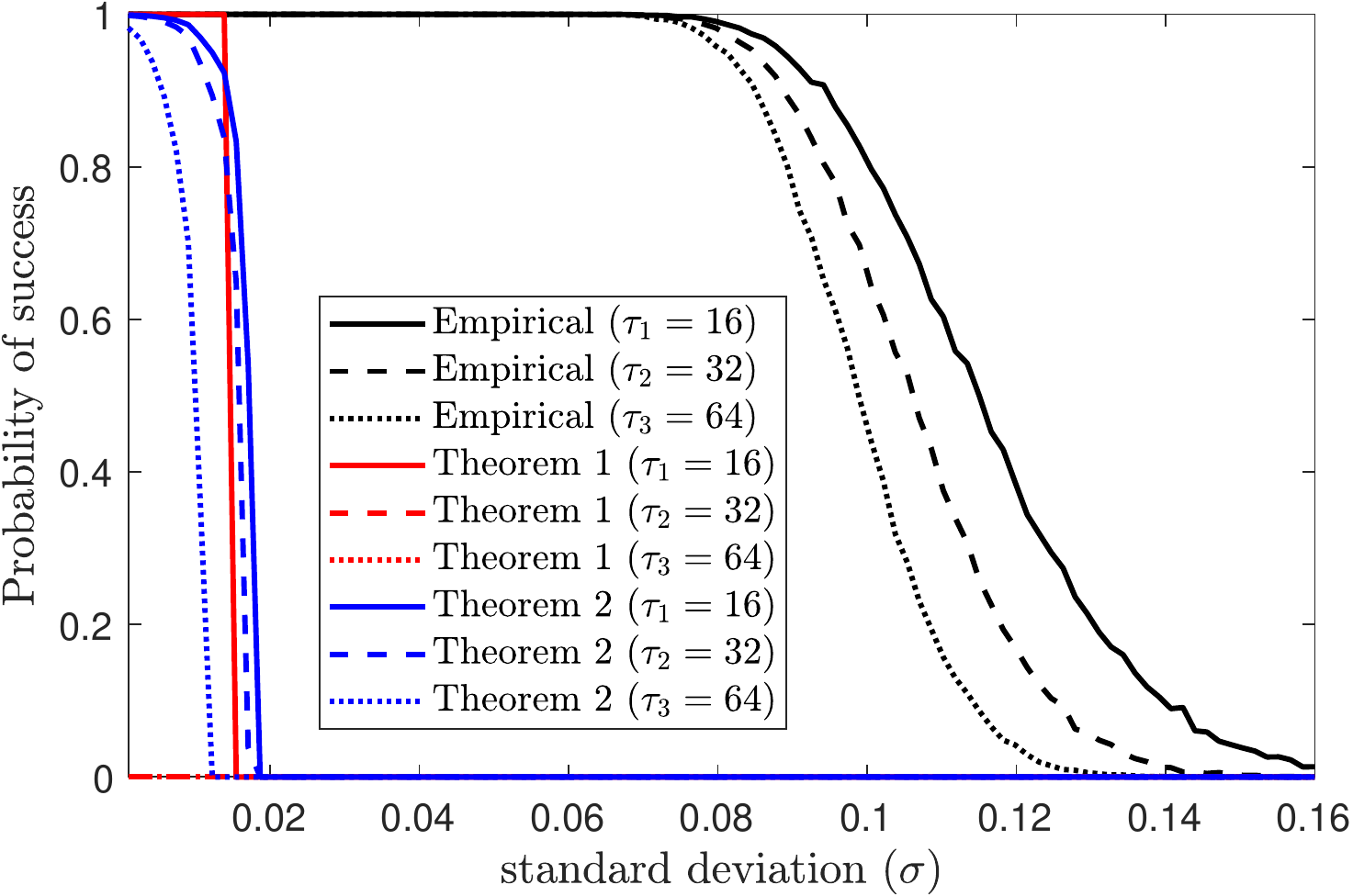}}
	\subcaptionbox{$M=4096$, $s_\mathrm{min}=0.5$, $s_\mathrm{max}=1$,\label{sigma-pr-4096}}{\includegraphics[width=0.95\linewidth]{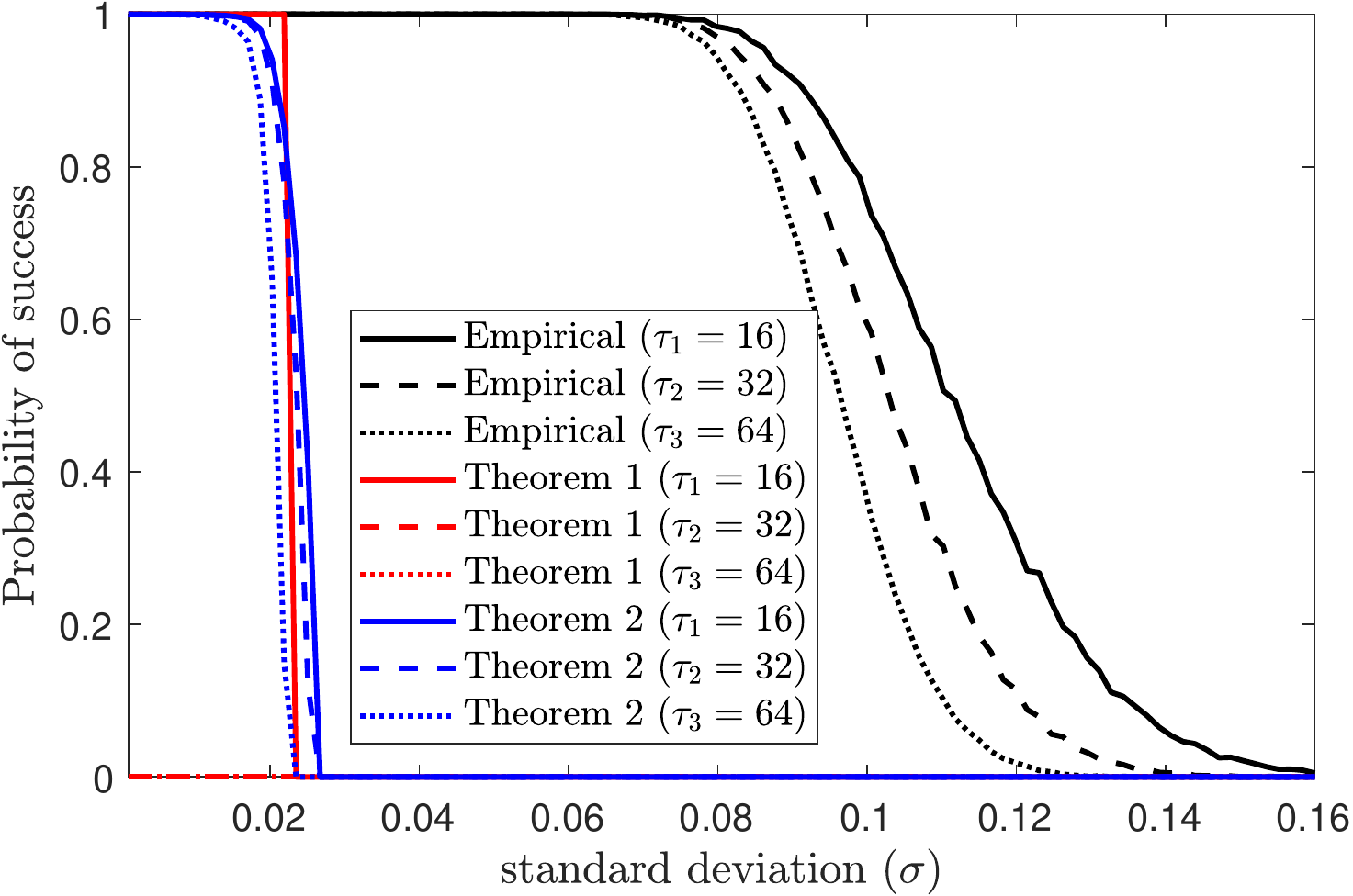}}
	\end{minipage}
	\caption{In each column of plots we demonstrate the effect of one parameter on the probability of successful support recovery while fixing the other parameters. Rows represent different values of $M$. The mutual coherence of the dictionary varies based on $M$. For $M=1024$, $M=2048$, and $M=4096$, we have $\mu_\mathrm{max}=0.0313$, $\mu_\mathrm{max}=0.0221$, and $\mu_\mathrm{max}=0.0156$, respectively.}
	\label{fig:pr-all}
\end{figure*}

All the empirical results are obtained by performing the OMP algorithm $5000$ times using a random sparse signal with additive white Gaussian noise in each trial. The probability of success is computed as the ratio of successful trials to the total number of trials; note that a trial is successful if $\Lambda=\hat{\Lambda}$, where $\hat{\Lambda}$ is the support of $\hat{\vect{s}}$ obtained from OMP by solving \eqref{eq:l0}. Moreover, the number of trials was empirically set such that the probability of success for the OMP algorithm was stable across different parameters. For comparison, we use the dictionary of \cite{ref30} defined as $\mat{A} = [\mat{I},\mat{H}]$, where $\mat{I}$ is an identity matrix and $\mat{H}$ is a Hadamard matrix, hence we have $N=2M$. 

The sparse signal in each trial, denoted $\vect{s}$ in \eqref{eq:problem:sp}, is constructed as follows: The support of the sparse signal, $\Lambda=\mathrm{supp}(\vect{s})$, is constructed by uniform random permutation of the set $\{1,\dots,N\}$ and taking the first $\tau$ indices. The nonzero elements located at $\Lambda$ are drawn randomly from a uniform distribution on the interval $[s_\mathrm{min},s_\mathrm{max}]$, multiplied randomly by $+1$ or $-1$. Once the sparse signal is constructed, the input of the OMP algorithm, $\vect{y}$, is obtained by evaluating \eqref{eq:problem:sp}. 

In order to facilitate the comparison of Theorems \ref{thm:elad} and \ref{thm:ours}, we need to fix the value of $\beta$. To do this, we empirically calculate $\beta$ as $\underset{\vect{w}}{\mathrm{max}} \; \underset{j}{\mathrm{max}}|\langle\mat{A}_j,\vect{w}\rangle|$, where the maximum over $\vect{w}$ is computed using $10^4$ vectors $\vect{w}\sim\mathcal{N}(0,\sigma^2\mat{I})$, as assumed by both theorems. Given $\beta$, we can calculate $\alpha$ for Theorem \ref{thm:elad} from the definition $\beta \triangleq \sigma \sqrt{2(1+\alpha)\log N}$. Indeed, a lower value of $\beta$ leads to better results for both theorems, see \eqref{eq:ours} and \eqref{eq:elad_cond}. As a result, here we consider the worst-case scenario. When \eqref{eq:elad_cond} is not satisfied for Theorem \ref{thm:elad}, we set the probability of success to zero. We use the same procedure for the condition of Theorem \ref{thm:ours}; i.e. the probability of success is set to zero when $s_\mathrm{min}/2<\beta$.%

Numerical results are summarized in Figure \ref{fig:pr-all}. We analyze the effect of sparsity on the probability of successful support recovery in plots \ref{tau-pr-1024}, \ref{tau-pr-2048}, and \ref{tau-pr-4096}. Three signal dimensionalities and three noise variances: $\sigma_1^2=10^{-6}$, $\sigma_2^2=2.5\times 10^{-5}$, and $\sigma_3^2=10^{-4}$, are considered. For all these cases we set $s_\mathrm{min}=0.5$ and $s_\mathrm{max}=1$. In Fig. \ref{tau-pr-1024} we see that Theorem \ref{thm:elad} achieves a higher probability for $\sigma_3$ and small values of $\tau$, while Theorem \ref{thm:ours} leads to more accurate results for larger values of $\tau$. Additionally, for $\sigma_1$ and $\sigma_2$, Theorem \ref{thm:ours} is much closer to empirical results. Most importantly, the shape of the probability curves for Theorem \ref{thm:ours} matches the empirical curves. In contrast, Theorem \ref{thm:elad} produces a step function due to the fact that condition \eqref{eq:elad_cond} is not satisfied for a large range of values for $\tau$, even though the success probability in \eqref{eq:elad_prob} is close to one for different values of $\sigma$. The condition of Theorem \ref{thm:ours} is satisfied across all the parameters for figures \ref{tau-pr-1024}-\ref{tau-pr-4096}. 

We discussed in section \ref{sec:analysis} that \eqref{eq:ours} is always smaller than \eqref{eq:elad_prob} due to the second term of \eqref{eq:ours}. We expect this term to become more accurate as the signal dimensionality grows since it is exponential in $N$; moreover, $\beta$ and $\mu_\mathrm{max}$ become smaller as $N$ grows. This is confirmed in figures \ref{tau-pr-2048} and \ref{tau-pr-4096}. As we increase $N$, the gap between theorems \ref{thm:elad} and \ref{thm:ours} increases, confirming that the second term of \eqref{eq:ours} is becoming more accurate compared to \eqref{eq:elad_cond}. The empirical probability is close to one for all the values of $\tau$ plotted in figures \ref{tau-pr-2048} and \ref{tau-pr-4096}. 

The effect of $s_\mathrm{min}$ on the probability of success is demonstrated in figures \ref{smin-pr-1024}, \ref{smin-pr-2048}, and \ref{smin-pr-4096}. For each plot, we consider $\tau_1=16$, $\tau_2=32$, and $\tau_3=64$, while setting $\sigma^2=10^{-4}$. The empirical results show a probability of success close to one across the parameters considered. In Fig. \ref{smin-pr-1024} we see a significant difference between Theorems \ref{thm:elad} and \ref{thm:ours}. The condition of Theorem \ref{thm:elad} is not satisfied for any value of $s_\mathrm{min}$ and $\tau$. In contrast, Theorem \ref{thm:ours} shows high probabilities for all three values of $\tau$. The dynamic range (DR) of the signal can be defined as $s^2_\mathrm{max}/s^2_\mathrm{min}$. As we increase the signal dimensionality ($N$), Theorem \ref{thm:ours} reports larger probability for larger values of DR and all three values of $\tau$. On the other hand, the condition of Theorem \ref{thm:elad} fails for $\tau_2$ and $\tau_3$, even when we have $M=4096$. For $\tau_1$, Theorem \ref{thm:elad} can produce valid results for a slightly higher DR. 

Lastly, in plots \ref{sigma-pr-1024}, \ref{sigma-pr-2048}, and \ref{sigma-pr-4096}, we analyze the effect of noise variance on the probability of success for $\tau_1=16$, $\tau_2=32$, and $\tau_3=64$. In Fig. \ref{sigma-pr-1024}, where $M=1024$, both theorems fail to produce valid results for $\tau_3=64$. However, Theorem \ref{thm:ours} reports acceptable results for $\tau_1$ and $\tau_2$, while the condition of Theorem \ref{thm:elad} is not satisfied. As the signal dimensionality grows, see Fig. \ref{sigma-pr-2048} and \ref{sigma-pr-4096}, Theorem \ref{thm:ours} becomes more tolerant of higher noise variances. The results for Theorem \ref{thm:elad} also improves with increasing signal dimensionality, however only for $\tau_1$. This shows the robustness of Theorem \ref{thm:ours} to larger values of sparsity. 
	




\section{Conclusions}  \label{sec:conclusion}
We presented a new bound for the probability of correctly identifying the support of a noisy sparse signal using the OMP algorithm. Compared to the analysis of Ben-Haim et al.~\cite{ref30}, our analysis replaces a sharp condition with a probabilistic bound. Comparisons to empirical results obtained by OMP show a much improved correlation than previous work.

\appendix
\begin{proof}[proof of Lemma \ref{lem:1}]
	Expanding $\Gamma_j$, we can show that 
	\begin{align}
	\Gamma_j &= \left|\sum_{m=1}^{M}\mat{A}_{m,j}\left(\sum_{n=1}^{N}\mat{A}_{m,n}\vect{s}_n+\vect{w}_m\right)\right|  \\
	&=\left|\sum_{n=1}^{N}\left\{\sum_{m=1}^{M}\mat{A}_{m,j}\mat{A}_{m,n}\vect{s}_n + \frac{1}{N}\sum_{m=1}^{M}\mat{A}_{m,j}\vect{w}_m\right\}\right|.  \\
	&=\left|\sum_{n=1}^{N}\left\{\mu_{j,n}\vect{s}_n+\frac{1}{N}\langle\mat{A}_j,\vect{w}\rangle\right\}\right|.
	\label{eq:lemma1_proof_1}
	\end{align}
	We are interested in tail bounds for sum of random variables $\mu_{j,n}\vect{s}_n+ N^{-1} \langle\mat{A}_j,\vect{w}\rangle$, for $n=1,\dots,N$. Let us define $\vect{x}_n = \mu_{j,n}\vect{s}_n$. Using the assumption $|\langle\mat{A}_j,\vect{w}\rangle|\le\beta$ we have
	\begin{align}	
	\mathrm{Pr} \left\{ \Gamma_j \ge \xi \right\} &\le \mathrm{Pr} \left\{\left|\sum_{n=1}^{N}\vect{x}_n\right| + \left|\frac{1}{N}\sum_{n=1}^{N} \langle\mat{A}_j,\vect{w}\rangle\right| \ge \xi  \right\} \nonumber \\
	&\le \mathrm{Pr} \left\{ \left|\sum_{n=1}^{N}\vect{x}_n\right| \ge \xi - \beta \right\}. \label{eq:lemma1_proof_4} 
	\end{align}
	Since $\{\vect{s}_n\}_{n=1}^N$,  and hence $\{\vect{x}_n\}_{n=1}^N$, are centered independent real random variables, according to Bernstein's inequality \cite{bernstein}, if  $\mathrm{E}\left\{\vect{x}_n^2\right\} \le \nu$, and $\mathrm{Pr}\{|\vect{x}_n| < c\}=1$, then for a positive constant $\delta$ we have
	\begin{align}
	\mathrm{Pr}\left\{\left|\sum\limits_{n=1}^N\vect{x}_n\right| \ge \delta \right\} &\le 2\;\mathrm{exp}\left(\frac{-\delta^2}{2\left(\sum\limits_{n=1}^{N}\mathrm{E}\left\{\vect{x}_n^2\right\}+c\delta/3\right)}\right) \nonumber \\
	&\le 2\;\mathrm{exp}\left(\frac{-\delta^2}{2(N\nu+c\delta/3)}\right). \label{eq:lemma1_proof_5}
	\end{align}
	\noindent Setting $\delta=\xi-\beta$ in \eqref{eq:lemma1_proof_5} completes the proof.
\end{proof}

\vfill
\newpage




%

\bibliographystyle{IEEEtran}
\bibliography{main}

\end{document}